\documentclass[lettersize,journal]{IEEEtran}

\usepackage{tabularx}

\ifCLASSINFOpdf
\else
\fi
\usepackage[numbers]{natbib}
\usepackage{epsfig}
\usepackage{epstopdf}
\usepackage[T1]{fontenc}
\usepackage{amssymb}
\usepackage{amsmath}
\usepackage{amsfonts}
\usepackage{verbatim}
\usepackage{graphicx}
\usepackage{hyperref}
\usepackage[usenames,dvipsnames]{xcolor}
\usepackage{lipsum}

\usepackage{multicol}
\usepackage{pgfplots}
\usepackage{lipsum}
\usepackage{colortbl}
\usepackage{algorithm}
\usepackage[noend]{algpseudocode}
\usepackage{enumitem}
\makeatletter
\def\algbackskip{\hskip-\ALG@thistlm}
\makeatother

\usepackage{multirow}

\usepackage{colortbl}
\usepackage{tabularx}
\usepackage{lipsum}  
\usepackage{pifont}
\usepackage{makecell}
\usepackage{float}
\usepackage{amsthm}
\usepackage{amsmath}
\usepackage{pgf}
\usepackage[normalem]{ulem}
\usepackage{wrapfig}

\newtheorem{lemma}{Lemma}

\newcommand\subsubsubsection{\@startsection{paragraph}{4}{\z@}%
  {1.5ex \@plus 1ex \@minus .2ex}%
  {-1em}%
  {\normalfont\normalsize\bfseries}}
\makeatother

\usepackage{tikz,bm,color}

\usetikzlibrary{circuits.logic.US,circuits.logic.IEC,fit}

\begin{document}
%
\title{Error Detection Schemes for Barrett Reduction of CT-BU on FPGA in Post Quantum Cryptography }

\author{
  \IEEEauthorblockN{Paresh Baidya\IEEEauthorrefmark{1}\IEEEauthorrefmark{3}, Rourab Paul\IEEEauthorrefmark{2}, Vikas Srivastava\IEEEauthorrefmark{4}, Sumit Kumar Debnath\IEEEauthorrefmark{1}}\\
     \IEEEauthorblockA{\IEEEauthorrefmark{1}Department of Mathematics, National Institute of Technology, Jamshedpur, India}\\
     \IEEEauthorblockA{\IEEEauthorrefmark{2}Department of Computer Science and Engineering, Shiv Nadar University, Chennai, Tamil Nadu, India}\\
   \IEEEauthorblockA{\IEEEauthorrefmark{3}Department of Computer Science and Engineering, Siksha ‘O’ Anusandhan Deemed to be University, Bhubaneswar, India\\
        \IEEEauthorblockA{\IEEEauthorrefmark{4}Department of Mathematics, Indian Institute of Technology Madras, Chennai, India}\\      
          \IEEEauthorblockA{\IEEEauthorrefmark{3}pareshbaidya@soa.ac.in}
           \IEEEauthorblockA{\IEEEauthorrefmark{2}rourabpaul@gmail.com}\\
   }
   
}

\maketitle

\begin{abstract}
A fault can occur naturally or intentionally. However, intentionally injecting faults into hardware accelerators of Post-Quantum Cryptographic (PQC) algorithms may leak sensitive information. This intentional fault injection in side-channel attacks compromises the reliability of PQC implementations. The recently NIST-standardized key encapsulation mechanism (KEM), Kyber may also leak information at the hardware implementation level.
This work proposes three efficient and lightweight recomputation-based fault detection methods for Barrett Reduction in the Cooley-Tukey Butterfly Unit (CT-BU) of Kyber on a Field Programmable Gate Array (FPGA). The CT-BU and Barrett Reduction are fundamental components in structured lattice-based PQC algorithms, including Kyber, NTRU, Falcon, CRYSTALS-Dilithium, etc.
This paper introduces a new algorithm, Recomputation with Swapped Operand (RESWO), for fault detection. While Recomputation with Negated Operand (RENO) and Recomputation with Shifted Operand (RESO) are existing methods used in other PQC hardware algorithms. To the best of our knowledge, RENO and RESO have never been used in Barrett Reduction before.
The proposed RESWO method consumes a similar number of slices compared to RENO and RESO. However, RESWO shows lesser delay compared to both RENO and RESO.
The fault detection efficiency of RESWO, RENO, and RESO is nearly $\sim$100\%.
\end{abstract}

\section{Introduction}
The rapid advancements in quantum computing present a significant threat to traditional public-key cryptographic systems (e.g., RSA \cite{mollin2002rsa} and elliptic curve cryptography (ECC)\cite{hankerson2004guide}). The security of most of the classical cryptographic schemes depends on the computational hardness of mathematical problems like integer factorization and discrete logarithms. However, these hard problems can be solved efficiently using Shor’s algorithm \cite{shor2002introduction} on a sufficiently powerful quantum computer. As a result, there is an urgent need to transition towards post-quantum cryptographic (PQC) \cite{overbeck2009code,srivastava2023overview,micciancio2009lattice} schemes that remain secure even in the presence of quantum adversaries.

The National Institute of Standards and Technology (NIST) initiated a post-quantum cryptography standardization process\cite{alagic2019status} to address this challenge in 2017. The aim is to identify cryptographic algorithms that can replace classical public-key cryptography. Among the various proposals, lattice-based cryptography emerged as a strong candidate due to its worst-case hardness guarantees, efficiency, and versatility. In 2024, NIST selected CRYSTALS-Kyber~\cite{bos2018crystals}, a lattice-based KEM, as the standard post-quantum KEM. In addition, two more lattice-based digital signature algorithms \cite{ducas2018crystals,prest2020falcon} were chosen for the standardization.

Although lattice-based cryptographic algorithms provide strong theoretical security guarantees, their practical implementations introduce several challenges, particularly in hardware-based deployments such as ASIC (Application-Specific Integrated Circuit), FPGA and Embedded Processors. Hardware implementations of PQC schemes are essential for high performance and efficiency. However, these implementations are vulnerable to various physical attacks, including side-channel and fault injection attacks. Side-channel attacks exploit unintended physical emissions such as power consumption, electromagnetic radiation, and timing information to extract cryptographic secrets. Fault injection attacks involve intentionally manipulating hardware operations to induce computational errors that potentially leak information about the secret keys.

Arithmetic operations like modular reduction are fundamental to structured lattice-based cryptographic schemes such as Kyber\cite{bos2018crystals}, NTRU\cite{hoffstein1998ntru}, Falcon\cite{prest2020falcon}, and CRYSTALS-Dilithium\cite{ducas2018crystals}. Modular reduction, particularly Barrett Reduction, is widely used in Number Theoretic Transform (NTT) computations. Both Dilithium and Kyber operate over the cyclotomic ring $\mathbb{Z}_q [x]/(x^n +1)$ utilizing the Number Theoretic 
Transformation (NTT) to accelerate the polynomial multiplication. The efficient implementation of Barrett Reduction is essential for maintaining high-speed and low-power cryptographic operations. However, the vulnerability of Barrett Reduction to fault injection attacks poses a significant security risk because even small perturbations in the computation can result in the leakage of sensitive information.

There is an increase in the deployment of PQC algorithms in hardware accelerators, particularly FPGAs. FPGAs offer flexibility and high performance for cryptographic implementations. However, their reconfigurable nature also makes them susceptible to various attacks, including transient faults induced by environmental factors and fault injection attacks using techniques such as voltage glitching and clock manipulation. Given the critical role of Barrett Reduction in lattice-based PQC schemes such as Kyber, there is a strong need to design an error detection scheme for Barrett reduction on FPGA to secure it against fault attacks.
\subsection{Literature}
Several research studies on FPGA, ASIC and Embedded Processor platforms have proposed recomputation and parallel computation techniques to address fault detection in various components of both classical and post-quantum cryptosystems.
In paper \cite{canto2}, Canto et al. implement fault detection hardware accelerators for lattice-based Key Encapsulation Mechanisms (KEMs) on a Kintex Ultrascale+ FPGA. They proposed three schemes: Re-computing with Shifted Operands (RESO), Re-computing with Negated Operands (RENO) and Re-computing with Scaled operands (RECO) for the Multiply-Accumulate (MAC) operation. This MAC computes $ACC = A \times B + C$ for matrix-matrix, matrix-vector, vector-vector, and polynomial multiplications. These RESO, RENO and RECO methods compute $ACC_{reso}$, $ACC_{reno}$ and $ACC_{reco}$ using equ. \ref{equ:reso}, \ref{equ:reno} and \ref{equ:reco} respectively.
\begin{equation}
\label{equ:reso}
 ACC_{reso} = \text{shift}_{2r}(\text{shift}_l(A) \times \text{shift}_l(B) + \text{shift}_{2l}(C))  
\end{equation}

\begin{equation}
\label{equ:reno}
 ACC_{reno} =-(-A \times B - C)
\end{equation}

\begin{equation}
\label{equ:reco}
 ACC_{reco} =\frac{t \times A \times t \times B + t^2 \times C}{t^2}
\end{equation}
Canto et al. \cite{canto2} compare $ACC$ with $ACC_{reso}$, $ACC_{reno}$, and $ACC_{reco}$ in the RESO, RENO, and RECO methods, respectively. If $ACC$ does not match $ACC_{reso}$ or $ACC_{reno}$, $ACC_{reco}$, it indicates a fault flag in the multiplier of the KEMs. These fault detection techniques are adopted in FrodoKEM, Saber, and NTRU which provide high error coverage with minimal performance overhead. It is important to note that Canto et al. \cite{canto2} only report the implementation overhead of RESO in FrodoKEM, Saber, and NTRU. They do not provide any overhead details for RECO and RENO. For instance, in the Saber MAC implementation, RESO consumes 19 Configurable Logic Blocks (CLBs) and 4.9 mW of power, whereas RENO and RECO consume 266 CLBs with 19 mW of power and 65 CLBs with 5.94 mW of power, respectively. Based on the reported RESO overhead, it can be reasonably assumed that the overhead values for RENO and RECO will be significantly higher.  \\
In \cite{kermani}, Kermani et al. propose an oblivious error detection scheme for Galois Counter Mode (GCM) on a 65nm ASIC platform. It is used to verify the integrity of data. The proposed approach improves the compatibility with various block ciphers and finite field multipliers. They use Re-computation of Swapped Cipher text and Additional authenticated Blocks (RESCAB) in the schemes. In this GCM, the Galois Hash (GHASh) is the main computation block which is computed on a GF[$2^{128}$]. On the other hand, the parallel RESCAB processes the swapped input with another $GF(2^{128})$. The outputs from GHASH and RESCAB are then compared to detect faults. This architecture improves design flexibility, as demonstrated through hardware implementations and error simulations. \\
A lightweight fault detection architecture for modular exponentiation $C=X^Y \mod n$, a crucial operation in both classical and post-quantum cryptography, is proposed in \cite{saeed} for FPGA implementations. The proposed Recomputation with Modular Offset (REMO) computes $C'=(X+Offset)^Y~mod~n$. The outputs from the modular exponentiation unit, $C$, and the REMO unit, $C'$, are then compared to detect faults. This method achieves nearly 100\% error detection with minimal computational and area overhead. \\
In \cite{dominguez}, a recomputation-based Point Validation (PV) method is employed for fault detection in elliptic curve scalar multiplication (ECSM). This method is implemented in Xilinx Virtex 2000E FPGA. It provides a high error coverage with minimal computational overhead. \\
In article \cite{sarker}, Sarker et al. implement a RENO model for Number Theoretic Transformation (NTT) in Zynq and Spartan7 FPGAs. This work implements three variants: $v1$, $v2$ and $v3$ based on the placement of RENO error checking in the logic path. Placing RENO deeper in the logic path increases slice utilization but increases error detection efficiency.
Ahmadi et al. \cite{ahmadECSM} present an algorithm for fault detection scheme for the window method in elliptic curve scalar multiplication (ECSM). In this paper, the authors propose an algorithm-level fault detection method  in window method scalar multiplication. Using simulation-based fault injection, they demonstrate that the scheme effectively detects a wide range of faults with high accuracy. This proposed method is implemented on both ARMv8 and FPGA architectures. Ahmadi et al. \cite{ahmadiNAF} also address a research gap in fault detection for $\tau$NAF conversion and Koblitz curve cryptosystems. Specifically, the authors introduce an algorithm-level fault detection scheme for the single $\tau$NAF conversion algorithm and two fault detection schemes for the double $\tau$NAF conversion algorithm. In this paper, the fault detection method is implemented on ARMv7 and ARMv8 architectures to evaluate its feasibility.\\
In the context of PQC schemes, Cintas et al. \cite{cintas} present an error detection schemes for Goppa arithmetic units used in the McEliece cryptosystem by utilizing the algebraic structure of underlying composite fields. They implement a Parity Checker for different sub components of McEliece. The schemes proposed in article \cite{cintas} are not only suitable for arithmetic units but are also applicable to core functions of other public-key cryptosystems that utilize composite fields as their mathematical base. The authors also provide the Goppa polynomial evaluation (GPE) implementations on an FPGA, and performance overheads are analyzed for different configurations. In the following, Canto et al. \cite{canto} introduce fault detection schemes for various finite-field operations, including addition, subtraction, multiplication, squaring, and inversion, within the context of the code-based McEliece cryptosystem. The authors implement the 5-bit Cyclic Redundancy Check (CRC-5) for different subcomponents of the McEliece cryptographic algorithm. The schemes proposed in \cite{canto} utilize different error detection techniques such as regular parity, interleaved parity, CRC-2, and CRC-8.  The proposed methods are integrated into distinct components of the Key Generator to enhance error detection probability, particularly in operations involving multiplications and inversions over \(GF(2^{13})\). Kamal et al. \cite{kamal} study various techniques to improve the fault resistance of NTRUEncrypt hardware implementations by proposing spatial duplication techniques. The proposed methods are evaluated based on their error detection capabilities, as well as their impact on area and throughput overheads.
\subsection{Our Claim} 
The aforementioned fault detection literature can be categorized into two types.
\begin{itemize}[noitemsep, left=25pt]
    \item[Type 1:] RESCAB \cite{kermani}, PV \cite{dominguez} fault detection methods which are very dependent on target crypt algorithms.
    \item[Type 2:] The fault detection methods Parity Checkers \cite{cintas}, CRC \cite{canto}, RENO \cite{sarker}, RESO \cite{canto2}, REMO \cite{saeed}, RECO \cite{canto2} are more general and easier to adopt as fault detection mechanisms in various cryptographic algorithms.
\end{itemize}

In our paper, we utilize RESO and RENO for the Barrett Reduction of the CT-BU, which falls under Type 2. We also benchmark a novel algorithm named RESWO for Barrett Reduction on FPGA, which falls under the Type 2 fault detection category.
To the best of our knowledge, this work is the first to propose recomputation-based error detection schemes for Barrett Reduction, which is one of the most resource and delay-intensive fundamental design blocks in many PQC algorithms including Round 3 finalists: Kyber, CRYSTALS-Dilithium, Falcon, and NTRU \cite{vadim}.
The primary contributions of the paper can be summarized as:
\begin{itemize}
			\item This paper proposes RESWO, a novel recomputation-based fault detection algorithm for Barrett Reduction in the CT-BU of NTT operations. The efficient FPGA implementation of RESWO maintains a similar slice overhead with reduced delay compared to existing solutions, while achieving a very high fault detection efficiency of 99.97\%. This makes RESWO suitable for high speed resource constrain hardware platforms. Thus, RESWO can be used in any polynomial multiplication with a modular reduction process.
			
			\item  To the best of our knowledge, this is the first work to implement and evaluate Recomputation with Negated Operand (RENO) and Recomputation with Shifted Operand (RESO) for Barrett Reduction in CT-BU within the NTT operation. A detailed comparative analysis between RESWO, RENO, and RESO reveals that while all three achieve similar fault detection efficiency, RESWO outperforms in delay.

			\item  The proposed RESWO is integrated into multiple NIST Round 3 PQC finalists, such as Kyber, CRYSTALS-Dilithium, Falcon, and NTRU, to validate its practical applicability. Our FPGA implementation results demonstrate that RESWO improves the security of these PQC schemes against fault attacks. RESWO keeps the area, power, and delay overheads minimal. In addition, error detection efficiency evaluations for random and burst fault injection (1-23 bit faults) confirm that RESWO consistently achieves $\sim$99.97\% detection accuracy.
\end{itemize}
The organization of the article is as follows: 
The proposed fault detection scheme is detailed in Section \ref{sec:method}, while the results are discussed in Section \ref{sec:res}. Finally, the conclusions are provided in Section \ref{sec:con}.
\vspace{-10pt}

\section{Fault Detection Methods}
\label{sec:method}
This paper employs three recomputation methods: RESWO, RENO and RESO to detect faults in Barrett Reduction used in CT-BU. The Barrett Reduction is the most resource-intensive, latency-critical, and energy-demanding operation in the NTT transformation. These three recomputation methods take encoded operands from the main Barrett Reduction unit and recompute the operations inside Barrett Reduction with a delayed clock input. The correlation between the intermediate values of different registers used in Barrett Reduction and those of the recomputation units helps detect both transient and permanent faults in the CT-BU.
In our recomputation methods and Barrett Reduction, instead of computing on all $l$ bits of the polynomial coefficient $\alpha(x)$ at once, these three fault detection methods operate on smaller, fixed word sizes of $w$ bits from the total $l$ bits of $\alpha(x)$, where $w \leq l$. This wordwise modification of Barrett Reduction is required for 2 reasons. 
(i) Instead of looking into the final values, comparing intermediate data in each loop may increase the fault detection efficiency. However, the study of $w$ vs. error detection efficiency in Table \ref{tab:fault_reswo} shows that $w$ has no effect on error detection efficiency. (ii) The adjustable $w$ allows tuning of power consumption, throughput and resource usage of the design which offers architectural flexibility.
\vspace{-10pt}

\subsection{Modified Barrett Reduction}
The proposed Modified Barrett Reduction for Fault Detection (MBRFD) method relies on the recomputation technique where the input $\alpha$ and $\beta$ are encoded to detect transient and permanent faults during modular reduction. As shown in Algorithm \ref{algo:mbrfd}, lines 4 and 5, the operands $\alpha$ and $\beta$ split into smaller word-sized components $(\alpha w_i, \beta w_j)$ to facilitate word-wise processing. Then, at line 6 of Algorithm \ref{algo:mbrfd}, it computes the intermediate product of word-sized operands $(\alpha w_i, \beta w_j)$ in each iteration. Line 7 appends $(i+j)\times w$ zeros to the $c$ to ensure the required bit-length constraints. The modular reduction then multiplies the quotient term with q and subtracts from c, where $\mu=\lfloor \frac{2^{2\times l}}{q} \rfloor$ is precomputed value. In line 9, Recomputation Unit ($ReComp$) computes an alternative remainder $r^f$. In the modular reduction process, the conditional statements in lines 11–15 ensure that the final result remains within the valid range. If the computed remainder $r$ is different from $r^f$, as shown in line 16, the fault flag $f_i$ is set to ‘1’. It indicates an error. Otherwise, $f_i$ remains ‘0’ i.e. it confirms fault-free execution.
  
 \begin{algorithm}[!htb]
    \caption{Modified Barrett Reduction for Fault Detection in Hardware : MBRFD($\alpha$, $\beta$, $q$)}
    \label{algo:mbrfd}
   \textbf{Input} $\alpha =(\alpha_{l-1},...\alpha_{1}, \alpha_{0})$,
      $\beta =(\beta_{l-1},...\beta{_1}, \beta_{0})$ \\
    $q =(q_{l-1},...q_{1}, q_{0})$ where $\mu=\lfloor \frac{2^{2\times l}}{q} \rfloor$ \& $k=2\times l$ \\
  \textbf{Output} $\rho, f$
    \begin{algorithmic}[1]
    \State $\rho=0$	
      \For{i=0 to $(\frac{l}{w}-1)~~$}
       \For{j=0 to $(\frac{l}{w}-1)~~$}
      \State $\alpha w_i=\alpha_{[iw+w-1...iw]}$
      \State $\beta w_j=\beta_{[jw+w-1...iw]}$
      \State $c=\alpha w_i \times \beta w_j$
      \State $c=c$ || $(i+j) \times w \{0\}$
      \State $r=c - (c \times \mu)_{[2k-1... k]}\times q $
       \State $r^f=ReComp(\alpha w_i, \beta w_j)$ 
      \If {($r>n$)} 
        \State $\rho=\rho+r-n$
     \Else
       \State $\rho=\rho+r$
     \EndIf
     
     \If {$\rho>n$} 
       \State $\rho=\rho-n$
     \EndIf
     
     
     \If {$r!=r^f$} 
     \State $f_i='1'$
     \Else
     \State $f_i='0'$
     \EndIf
     
      \EndFor    
      \EndFor   
    \State  \textbf{return}   $\rho$, $f$
    \end{algorithmic}  
    \end{algorithm}
 \subsection{Recomputation Units}
This paper implements three $ReComp$ units: $RESWO$, $RENO$, and $RESO$, which run in parallel with the baseline PQC algorithms.
\subsubsection{RESWO}
\label{sec:reswo}
This paper proposes a new Recomputation algorithm named $RESWO$ to detect faults inside Barrett Reduction, as shown in Algorithm \ref{algo:reswo}. The proposed $RESWO$ algorithm takes swapped input words and adjusts the final output by multiplying with a value $\Delta$, which depends on the positions of the swapped inputs.
If we divide $l$ bits inputs: $\alpha$ and $\beta$ in $w$ bits word-size (segments), each word of $\alpha$ and $\beta$ can be expressed as: $\alpha w_i=\alpha_{[{iw+w-1}...i' . . .j'..{iw}]}$ and $\beta w_j=\beta_{[{iw+w-1}.. ...{iw}]}$. The swapped word of $\alpha$ is represented as ${\alpha w_i}^{swapped} = \alpha_{[{iw+w-1}..j' ..i'..{iw}]}$, where $q$ is the modulus. The ${\alpha w_i}^{swapped}$ denotes the value of $\alpha w_i$ with the $i'^{th}$ and $j'^{th}$ bits swapped. Here, $i' \leq w$ and $j'\leq w$.

 In Lemma \ref{lemma:1}, we prove that the multiplication of two word-size operands, $\alpha w_i$ and $\beta w_j$, is equal to $({\alpha w_i}^{swapped} \times \beta w_j)+\Delta \times \beta w_j$. If this holds true, then the zero-padded $c^f$ (line 5 of Algorithm \ref{algo:reswo}) is equal to $c$ (line 7 of Algorithm \ref{algo:mbrfd}), and consequently, $r^f$ (line 6 of Algorithm \ref{algo:reswo}) is equal to $r$ (line 8 of Algorithm \ref{algo:mbrfd}). In this paper, zero-padded $c^f$ is used in algorithms \ref{algo:reswo}, \ref{algo:reno}, \ref{algo:reso} on lines 5, 2, 2 respectively for the same logic.
\begin{algorithm}[!htb]
\caption{Recomputation with Swapped Operands (RESWO)}
\label{algo:reswo}
\textbf{Input} Two positive integers $\alpha$ and $\beta$ \\
\textbf{Output} \textbf{$r^f$}= $\alpha \times \beta~mod~q$
\begin{algorithmic}[1]
\State $\alpha_{swapped}$ = Swap bits at arbitrary position $i'$ and $j'$ in $\alpha$
\State Compute the difference $(\delta)$ between the two swapped bits in $\alpha$ at positions $i'$ and $j'$, i.e $\delta = \alpha[i']-\alpha[j'] $
\State Compute the weighted difference: $\Delta = \delta \times (2^{i'}-2^{j'})$
\State Compute the product: $c^f$= $(\alpha_{swapped} \times \beta)+ \Delta \times \beta $   
\State *Pad Zeros $c^f=c^f$ || $(i+j) \times w \{0\}$
\State Compute Remainder $r^f=c^f - (c^f \times \mu)_{[2k-1... k]}\times q $
\State \textbf{return} $r^f$
\end{algorithmic}
\end{algorithm}

\begin{lemma}
\label{lemma:1}
 If $\alpha w_i=\alpha_{[{iw+w-1}...i' . . .j'..{iw}]}$ and $\beta w_j=\beta_{[{jw+w-1}.. ...{jw}]}$ are the encoded  word segments of $\alpha$ and $\beta$ respectively, then $({\alpha w_i}^{swapped} \times \beta w_j)+\Delta \times \beta w_j$=$\alpha w_i \times \beta w_j$. 
\end{lemma}

\begin{proof}
Let $\alpha w_i$ be represented in binary as $\alpha w_i=\sum_{k=iw}^{iw+w-1} \alpha w_{k}.2^{k} $, Where, $\alpha w_{k}$ represents the bit at position k (either 0 or 1), and $2^k$ is the corresponding weight. Now, swap two arbitrary bits i' and j' in the word segment $\alpha w_i$ and produce new value ${\alpha w_i}^{swapped}=\alpha_{[{iw+w-1}..j' ..i'.\alpha'_{iw}]}$. Then, we have three results:
\begin{itemize}
    \item if $\alpha w_i[i']=1$ and $\alpha w_i[j']=0$, the swap decreases by $\alpha w_i$ by $2^{i'}-2^{j'}$.
    \item if $\alpha w_i[i']=0$ and $\alpha w_i[j']=1$, the swap increases by $\alpha w_i$ by $2^{i'}-2^{j'}$.
    \item if $\alpha w_i[i']=\alpha w_i[j']$, then swap has no changes.
\end{itemize}
Therefor, 
\begin{align*}
\delta w_i&=(\alpha w_i[i']-\alpha w_i[j'])
\end{align*}
 where, $\delta w_i$ is the  is the difference between the original bit values.
Now,
\begin{align*}
\Delta  w_i&=(\alpha w_i[i']-\alpha w_i[j'])\times(2^{i'}-2^{j'}).
\end{align*}
Then the new value of $\alpha w_i$ after the swap is: 
\begin{align*}
{\alpha w_i}^{swapped}&=\alpha w_i - \Delta w_i
\end{align*}
Therefore,
\begin{align*}
&{\alpha w_i}^{swapped}\times \beta w_j=(\alpha w_i - \Delta w_i) \times\beta w_j
\end{align*}
Now,
\small{
 ${(\alpha w_i}^{swapped}\times \beta w_j)+ \Delta w_i \times \beta w_j\\
=(\alpha w_i - \Delta w_i)\times \beta w_j + \Delta w_i \times \beta w_j$\\
$=\alpha w_i \times\beta w_j$
}
\end{proof}
\begin{figure}[!htb]
\centering
\includegraphics[width=0.5\textwidth]{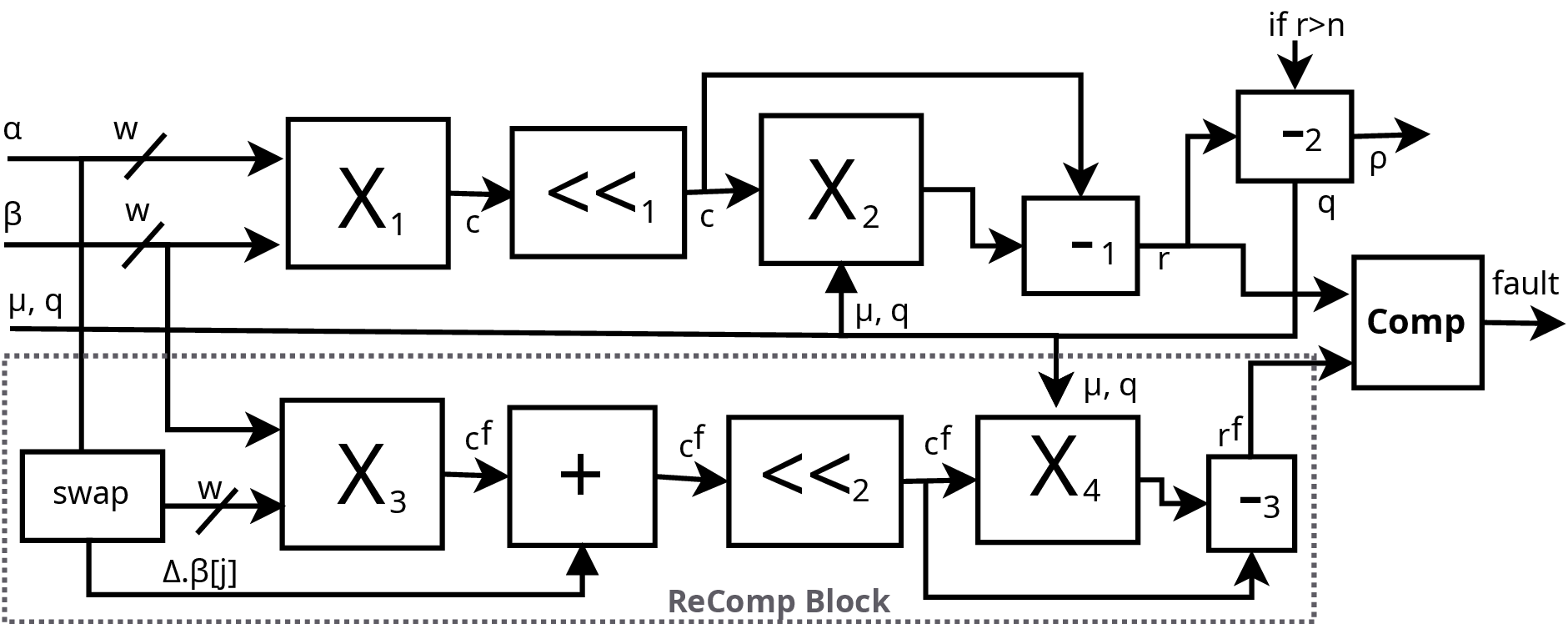}
\vspace{-5pt}
\caption{Hardware Architecture of Barrett Reduction with RESWO as ReComp Unit}
\vspace{-5pt}
\label{fig:reswo}
\end{figure}
As shown in Fig. \ref{fig:reswo}, the hardware of $RESWO$ module consists of two multipliers ($X_3$ and $X_4$), an adder ($+$), a subtractor ($-_3$), a left shifter ($<<_2$) and a swap block. 
The multiplier $X_3$ computes the product of ${\alpha w_i}^{swapped}$ and $\beta w_j$ where ${\alpha w_i}^{swapped}$ is generated by $swap$ block. Additionally, the $swap$ block computes the product of $\Delta$ and $\beta w_j$. The adder block $+$ add ${\alpha w_i}^{swapped}.\beta w_j$ with $\Delta.\beta w_j$ in $c_f$ by following line 4 of Algorithm \ref{algo:reswo}. The left shifter block $<<_2$ left shifts $c^f$ by $(i+j)w$ positions and pads $(i+j)w$ zero bits at the least significant bit position (line 5 of Algorithm \ref{algo:reswo}). 
The multiplier $X_4$ operates in two steps. In the first step, it computes the product of $c^f$ and $\mu$, extracting only the bits from positions $[2k-1... k]$, represented as $(c^f \times \mu){[2k-1... k]}$. In the second step, $X_4$ multiplies $(c^f \times \mu){[2k-1... k]}$ by $q$. As shown in line 6 of Algorithm \ref{algo:reswo}, the subtractor block $-_3$ subtract ($c^f \times \mu)_{[2k-1... k]}\times q$ from $c^f$ and finally produces $r^f$ to the $comp$ block for fault detection.
$c^f$ and the subtractor block $c^f$ from shifted $c^f$. The swap block exchanges the $i'^{th}$ and $j'^{th}$ position, producing ${\alpha w_i}^{swapped}$ (line 1 of \ref{algo:reswo}). Additionally, it computes $\Delta$ as shown in line 3 of \ref{algo:reswo}.

\subsubsection{RENO}
As shown in Fig. \ref{fig:reno}, the hardware of the $RENO$ module consists of two multipliers ($X_3$ and $X_4$), two 2’s complement blocks ($2'scompl_1$ and $2'scompl_2$), a left shifter ($<<_1$), and a subtractor ($-_3$).
This $ReComp$ $RENO$ variant takes $-\alpha$ instead of $\alpha$ from the $2's~compl_1$ block, as shown in Fig. \ref{fig:reno}. Consequently, line 1 of Algorithm \ref{algo:reno} computes $c^f=-\alpha \times \beta $ using the $X_3$ multiplier. Thereafter, the left shifter block $<<_2$ left shifts $c^f$ by $(i+j)w$ positions and pads $(i+j)w$ zero bits at the least significant bit position. Therefore line 2 of Algorithm \ref{algo:reno} concludes $c^f=-\alpha w_i \times \beta w_j$ || $(i+j) \times w \{0\}$. Subsequently, if we replace this $c^f$ in line 3,of Algorithm \ref{algo:reno}, it becomes:
\begin{align}
    r^f=&(\underbrace{-\alpha \times \beta || (i+j) \times w \{0\}}_{-c^f}) \\
    &- ((\underbrace{-\alpha \times \beta  || (i+j) \times w \{0\}}_{-c^f}) \times \mu)_{[2k-1... k]}\times q \nonumber
\end{align}
The above line can be simplified as : 
\begin{align}
 \label{eq1}
    r^f=-c^f -(- c^f \times \mu)_{[2k-1... k]}\times q )
\end{align}
\begin{algorithm}[!htb]
\caption{Recomputation with Negate Operands (RENO)}
\label{algo:reno}
\textbf{Input} Two positive integers $\alpha$ and $\beta$ \\
\textbf{Output} \textbf{$r^f$}= $\alpha \times \beta~mod~q$
\begin{algorithmic}[1]
\State Compute the product: $c^f$= $-\alpha \times \beta  $   
\State *Pad Zeros $c^f=c^f$ || $(i+j) \times w \{0\}$ 
\State Compute Remainder $r^f=-c^f - (-c^f \times \mu_{[2k-1... k]})\times q $
\State Compute 2's complement $r^f=-2's compl(r^f) $
\State \textbf{return} $r^f$
\end{algorithmic}
\end{algorithm}

The multiplier $X_4$ follows the same two-step multiplication process on $c^f$, $\mu$ and $q$ as described in Sec. \ref{sec:reswo}.
Then, the subtractor block ($-_3$) subtract $(-c^f \times \mu)_{[2k-1... k]}\times q$ form $c^f$. Finally, the $2's~compl_2$ block takes $-c^f - (-c^f \times \mu)_{[2k-1... k]} \times q$ as input to compute $c^f - (c^f \times \mu)_{[2k-1... k]} \times q$. The $comp$ module takes one input from $-_1$ and another from subtractor $-_3$. If the outputs of subtractor $-_1$ and subtractor $-_3$ differ, the $fault$ signal is set to $1$; otherwise, it remains $0$.
\begin{figure}[!htb]
\centering
\vspace{-5pt}
\includegraphics[width=0.5\textwidth]{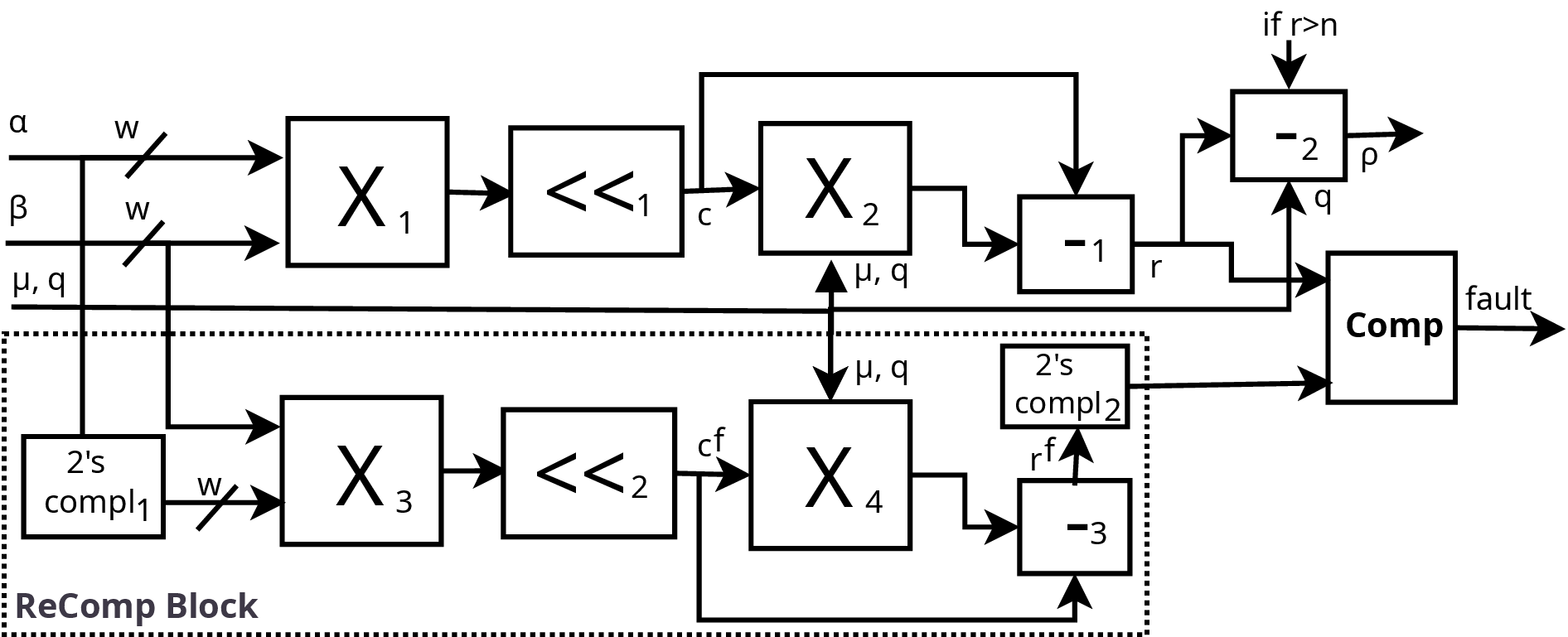}
\vspace{-5pt}
\caption{Hardware Architecture of Barret Reduction with RENO as ReComp Unit}
\vspace{-5pt}
\label{fig:reno}
\end{figure}
\subsubsection{RESO}
As shown in Fig. \ref{fig:reso}, the hardware of the $RESO$ module consists of two multipliers ($X_3$ and $X_4$), two left shifter ($<<_3$ and $<<_4$), a right shifter ($>>$) and a subtractor ($-_3$).
\begin{figure}[!htb]
\centering
\vspace{-5pt}
\includegraphics[width=0.5\textwidth]{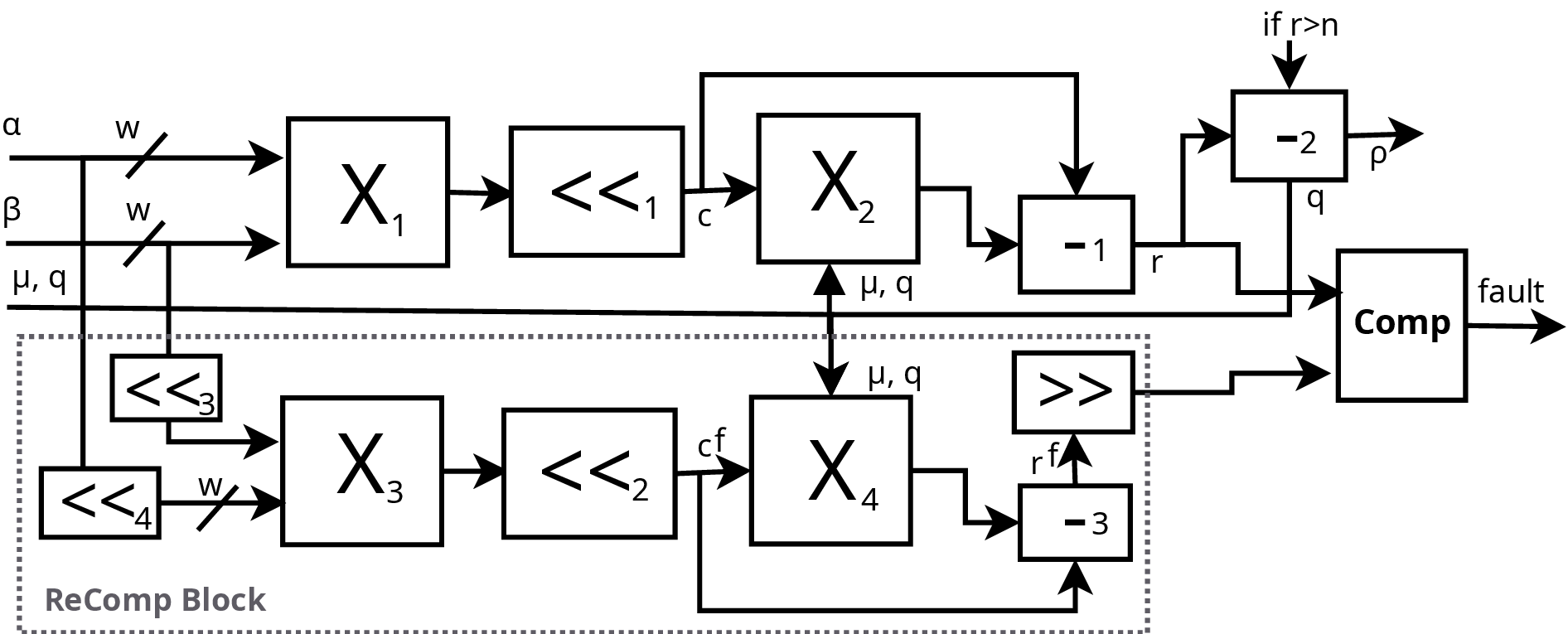}
\vspace{-5pt}
\caption{Hardware Architecture of Barrett Reduction with RESO as ReComp Unit}
\vspace{-5pt}
\label{fig:reso}
\end{figure}
As stated in line 1 of Algorithm \ref{algo:reso}, the left shifters $<<_3$ and $<<_4$ shift $\alpha$ and $\beta $ to the left and pad a $0$ at the least significant position of $\alpha$ and $\beta$, respectively. This one bit left shifted $\alpha$ and $\beta$ is presented as $shift_l(\alpha)$ and $shift_l(\beta) $ respectively. Therefore, the $X_3$ multiplier computes the product of $shift_l(\alpha)$  and $shift_l(\beta)$ in $c^f$ (line 1 of Algorithm \ref{algo:reso}). Consequently, left shifter block $<<_2$ left shifts $c^f$ by $(i+j)w$ positions and pads $(i+j)w$ zero bits at the least significant bit position (line 2 of Algorithm \ref{algo:reso}). 
\begin{algorithm}[!htb]
\caption{Recomputation with Shift Operands (RESO)}
\label{algo:reso}
\textbf{Input} Two positive integers $\alpha$ and $\beta$ \\
\textbf{Output} \textbf{$r^f$}= $\alpha \times \beta~mod~q$
\begin{algorithmic}[1]
\State Compute the product: $c^f$= $shift_l(\alpha) \times shift_l(\beta) $   
\State *Pad Zeros $c^f=c^f$ || $(i+j) \times w \{0\}$
\State Compute $r^f=c^f - (c^f \times \mu_{[2k-1+2... k+2]})\times q $
\State Compute remainder $r^f=shift_r(r^f)$
\State \textbf{return} $r^f$
\end{algorithmic}
\end{algorithm}
In this RESO Algorithm, the input words $\alpha$ and $\beta$ of the $X_3$ multiplier are padded with an extra zero at least significant bit position, therefore, the computation of $c^f$ at line 8 of Algorithm \ref{algo:mbrfd} is changed. The conventional Barrett Reduction method takes bits from position $k$ to $2k-1$ from the left side of the product of $c_f$ and $\mu$. However, the RESO Barrett Reduction takes bits from position $k+2$ to $2k-1+2$ to compensate for the impact of padded zeros in the input words $\alpha$ and $\beta$ (line 2 of Algorithm \ref{algo:reso}). The multiplier $X_4$ follows the same two-step multiplication process on $c^f$, $\mu$ and $q$ as described in Sec. \ref{sec:reswo} except the bit position of the product of $c^f$ and $\mu$. The modified line 8 of Algorithm \ref{algo:mbrfd} in RESO is shown in line 3 of Algorithm \ref{algo:reso}.

The right shifter block $>>$ then shifts the final $r^f$ by two bits, as described in line 4 of Algorithm \ref{algo:reso} and sends it to the $comp$ block for fault detection.

\subsection{Hardware Architecture of CT-BU \& MBRFD}
Our CT-BU component is designed with three pipeline stages: 
\begin{itemize}
    \item The $1^{st}$ pipeline stage buffer  $r=\omega[m+i]$ and $U=\alpha[j+k]$ (line 7 and 9 of Algorithm \ref{algo:CTNTT}).
    \item The $2^{nd}$ pipeline stage compute $V=$MBRFD$(\alpha[j+t], r, q)$ (line 10 of Algorithm \ref{algo:CTNTT}).
    \item The $3^{rd}$ pipeline stage compute $\overline{\alpha}[j]=U+V$ and  $\overline{\alpha}[j+t]=U-V$ (line 11 and 12 of Algorithm \ref{algo:CTNTT}).
\end{itemize} 
where $\alpha$ is the input polynomial and $\overline{\alpha}$ is pointwise representation of of $\alpha$.
The details of CT-BU is stated in \cite{ahmet}.
\begin{algorithm}[!htb]
\caption{Cooley-Tukey Iterative NTT algorithm\cite{ctntt}}
\label{algo:CTNTT}
\textbf{Input} A vector $\boldsymbol{\alpha} =(\alpha_{n-1},...\alpha_{1}, \alpha_{0})$, where each $\alpha_i \in [0, q-1]$
      of degree $n$ (a power of 2) and modulus $q\equiv 1~mod~2n$ \\
\textbf{Input}  Precomputed table of 2n-th roots of unity $\omega$, in bit reversed order\\
\textbf{Output} $\boldsymbol{\bar {\alpha} }\leftarrow NTT(\boldsymbol{\alpha} )$  $\in \mathbb{Z}_q [x]/(x^n +1)$
\begin{algorithmic}[1]
\Function{NTT}{$\boldsymbol{\alpha} $}
    \State $t \gets n/2$
    \State $m \gets 1$
    \While{$m < n$}
        \State $k \gets 0$
        \For{$i \gets 0$ to $m-1$}
            \State $r \gets \omega[m + i]$
            \For{$j \gets k$ to $k + t - 1$}
                \State $U \gets \boldsymbol{\alpha}[j]$
                \State $V \gets \textbf{MBRFD}(\boldsymbol{\alpha}[j+t], r,q)$ 
                 
                \State $\boldsymbol{\overline{\alpha}}[j] \gets (U + V) \mod q$
                \State $\boldsymbol{\overline{\alpha}}[j + t] \gets (U - V) \mod q$
            \EndFor
            \State $k \gets k + 2t$
        \EndFor
        \State $t \gets t/2$
        \State $m \gets 2m$
    \EndWhile
    \State \Return 
\EndFunction
\end{algorithmic}
\end{algorithm}
The polynomial coefficients stored in the Polynomial Coefficient memory are accessed by various computation units, including the polynomial multiplier, NTT and polynomial adder. As shown in Fig. \ref{fig:ct_bf_arch}, the $din$, $addr$ and $rd\_wr\_en$ of the Polynomial Coefficient memory can be accessed by different computation units through the $mux$es, whose selecting inputs are controlled by $Control~Unit$. The $demux$es are used to read  $\alpha$ from polynomial coefficient memory through $dout$ by different computation units of Kyber. 
The $\alpha$ stored polynomial coefficient memory is sent to the $CT-BU$ for generating the $U$ and $V$. The $MBRFD$ block shown in Fig. \ref{fig:ct_bf_arch} computes $V$ following algorithm \ref{algo:mbrfd}. The $MBRFD$ block performs the multiplication of $\alpha$ and $\omega$ in a $w$ word-wise manner.
As shown in Fig. \ref{fig:ct_bf_arch}, the $Barrett$ $Reduction$ block takes $w$ bits named $\alpha w_i$ at a time from the $\alpha$. Similarly, the $ReComp$ block shown in Fig. \ref{fig:ct_bf_arch}, processes $w$ bits named $\alpha w_i^f$ at a time from the $\alpha$. 
The $Barrett$ $Reduction$ block compute $r$ and the $ReComp$ unit compute $r^f$. 
Then, the $comp$ block compares $r$ and $r^f$. If both $r$ and $r^f$ match, $Fault=0$ and $V$ is accepted to calculate $\overline{\alpha}[j]$ and $\overline{\alpha}[j+t]$. Otherwise, if a fault is detected ($Fault=1$), the $\overline{\alpha}[j]$ and $\overline{\alpha}[j+t]$ are computed by the $adder$ and $sub$ blocks, respectively, as shown in Fig. \ref{fig:ct_bf_arch}.

\begin{figure*}[!htbp]
\centering
\vspace{-5pt}
\includegraphics[width=0.7\textwidth]{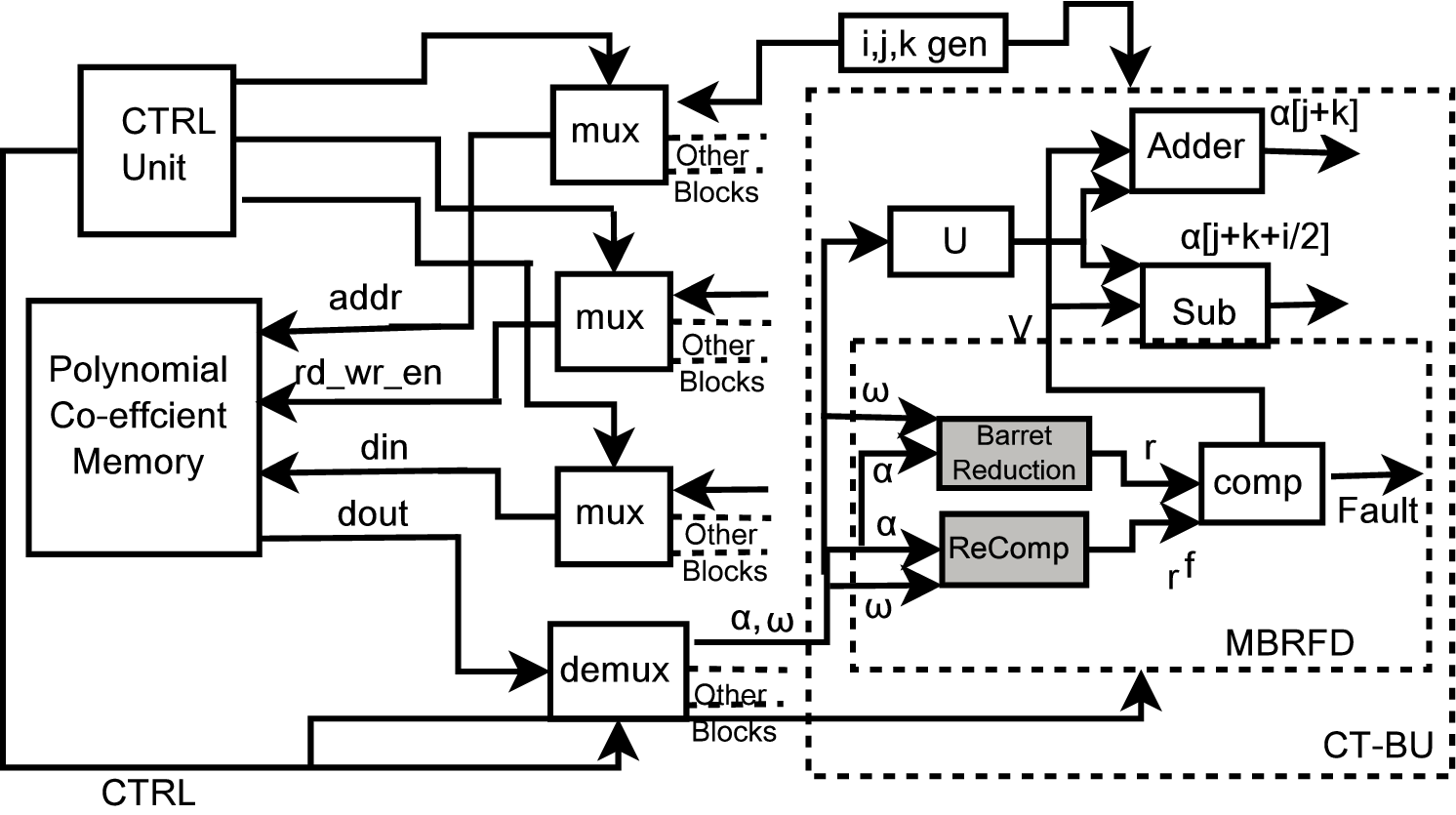}
\vspace{-5pt}
\caption{CT-BU Architecture With Fault Detection}
\vspace{-10pt}
\label{fig:ct_bf_arch}
\end{figure*}
\section{Results \& Discussions}
\label{sec:res}
This section discusses the overheads and error coverage of RESWO, RENO, and RESO compared to other existing fault detection solutions.
\begin{table*}[htbp!]
			\centering
			\begin{tabular}{|>{\centering\arraybackslash}p{4cm}|>{\centering\arraybackslash}p{1.5cm}|>{\centering\arraybackslash}p{1.5cm}|>{\centering\arraybackslash}p{1.5cm}|>{\centering\arraybackslash}p{1.5cm}|>{\centering\arraybackslash}p{1.5cm}|}
				\hline
				\textbf{Architecture} & \textbf{n, q} & \textbf{Slices} & \textbf{LUTs} & \textbf{FFs} & \textbf{Power (mW)}  \\ \hline\hline
				
				\textbf{Kyber CT-BU (Baseline)} & 256,  &573 & 972 & 239 & 131  \\ \cline{1-1} \cline{3-6}				
				\textbf{Kyber Barrett} & 3329  &76 & 254 & 89 & 99  \\ \cline{1-1} \cline{3-6}
				\textbf{Kyber Barrett with RESWO} &  & 128 & 444 & 127 & 101  \\
				\hline \hline
				
				\textbf{CRYSTALS-Dilithium CT-BU (Baseline)}   & 256,    &401 &689  &309 & 116  \\ \cline{1-1} \cline{3-6}
				\textbf{CRYSTALS-Dilithium Barrett}   & 8380417   &52 &127  &115 & 98  \\ \cline{1-1} \cline{3-6}
				\textbf{CRYSTALS-Dilithium Barrett with RESWO} & &77&175 &150&99   \\
				\hline\hline	
				\textbf{Falcon CT-BU (Baseline)}    & 512,  &314 &408&209&111   \\ \cline{1-1} \cline{3-6}
				\textbf{Falcon Barrett}    & 12289  &38 &88&73&96   \\ \cline{1-1} \cline{3-6}
				\textbf{Falcon Barrett with RESWO} &  &50&132&108&97  \\ \hline\hline	
				
				\textbf{NTRU CT-BU (Baseline) \cite{vadim}}      & 2048, & 312 & 411 & 207 & 111 \\ \cline{1-1} \cline{3-6}
				\textbf{NTRU Barrett}      & 12289 & 33 & 88 & 73 & 100 \\ \cline{1-1} \cline{3-6}
				\textbf{NTRU Barrett with RESWO} &  & 47 & 132 & 108 & 101 \\ \hline\hline
				
				\multicolumn{6}{|c|}{\textbf{Artix-7 (xc7a100tcsg324-3), w=4, clock=100MHz}} \\ \hline
			\end{tabular}
			\vspace{2pt}
			\caption{Overhead of RESWO in Different PQC Algorithms}
			\label{tab:compalgo}
\end{table*}
\subsection{Overheads}
The designs are becnhmarked on an $Artix$-$7$ $(xc7a100tcsg324$-$3)$ FPGA with $Vivado$ $22.02$ and the VHDL.
The overheads of the proposed fault detection algorithms such as RESWO, RENO and RESO are calculated from three perspectives:
\subsubsection{Overheads of RESWO in PQC Algorithms}
This paper implements Barrett Reduction with proposed RESWO fault detection model for the Kyber, CRYSTALS-Dilithium, Falcon and NTRU standards.
 \begin{table*}[!htbp]
 		\centering
	\begin{tabular}{|p{0.7cm}|>{\centering\arraybackslash}p{5.2cm}|>{\centering\arraybackslash}p{3.2cm}|>{\centering\arraybackslash}p{1.3cm}|>{\centering\arraybackslash}p{1.2cm}|>{\centering\arraybackslash}p{1.2cm}|>{\centering\arraybackslash}p{1.3cm}|}
	\hline
	\textbf{Work}& \textbf{Type of Fault } & \textbf{Platform}& \multicolumn{3}{c|}{\textbf{Overhead (\%)}}& \textbf{(\%) Error} \\ \cline{4-6}
	
	& \textbf{Detection \& Target HW} &  &\textbf{Area} & \textbf{Delay} & \textbf{Energy} &  \textbf{Coverage} \\ \hline
    	\cite{canto2}  &RESO (Saber/NTRU/ FrodoKEM) &Kintex
Ultrascale+ FPGA &36.6/39.6/ 28.4  & 28.3/16.7/ 32.7   & 1.2/3.2/ $\sim$0 & >99.9 \\ \hline
    	\cite{canto2}  &RECO \& RENO (Saber/NTRU/ FrodoKEM) &Kintex
Ultrascale+ FPGA &NA/NA/NA  & NA/NA/NA   & NA/NA/NA $\sim$0 & >99.9 \\ \hline
    \cite{kermani}  &RESCAB (Galois Counter Mode) & 65nm ASIC&4.9/6.7   & NA   & NA  & 100 \\ \hline
	\cite{saeed}  &REMO ($x^y~mod~n$) &Artix UltraScale+ FPGA& 0.8   & 0.27   & 0.65  & 97.1-100 \\ \hline
   \cite{dominguez}  &Point Validation (ECSM) & Spartan 3 1000 FPGA &15.17   & 4.8   & NA  & $\sim$99.99 \\ \hline    
    	\cite{sarker}  &RENO v1/v2/v3 (NTT) &Spartan7 FPGA& 20.2/15.3/ 21.5 *  & 8.46/15.88/ 13.71   & 15.6/7.6/ 11.2  & 99.51/99.67 /99.41 \\ \hline
        \cite{sarker}  &RENO v1/v2/v3 (NTT)& Zynq FPGA& 24/7.5/ 17*  & 9.32/19.66/ 21.78   & 20.47/13.27 /17.26  & 99.51/99.67 /99.41 \\ \hline
        \cite{ahmadECSM}  & Window Method Scalar Multiplication (ECSM) &ZYNQ Ultrascale+ FPGA&  1.8 &  0  & 0.1&39-99.9  \\ \hline

    \cite{ahmadiNAF}  & Coherency Check(Single $\tau$NAF) &ARM CORTEX-M4 Processor&  - &  8.5  & NA&83-97  \\ \hline
	\cite{cintas}  &1/2/3-bit parity (McEliece) &Kintex-7 FPGA& 9.8/11.3/9.6   & 1.4/0.8/1   & 2.7/2.7/2.7  & 100 \\ \hline        
	\cite{canto}  &CRC5 (sub, add of McEliece) & Kintex-7 FPGA & 18.33   & 11.25   & $\sim$0  & >99.9 \\ \hline
	\cite{kamal}  &Spatial duplication (NTRU)& Virtex-E FPGA& 6.22   & NA   & NA  & 100 \\ \hline

    \textbf{Our} & \textbf{RESWO (CT-BU of Kyber)} &Artix-7 FPGA& \textbf{9.07}  & \textbf{2.02} &\textbf{1.52}& \textbf{$\sim 99.97$} \\ \hline
    \textbf{Our} & \textbf{RENO (CT-BU of Kyber)} &Artix-7 FPGA& \textbf{9.77}   & \textbf{2.98} &\textbf{1.52}& \textbf{$\sim 99.97$} \\ \hline
    \textbf{Our} & \textbf{RESO (CT-BU of Kyber)}& Artix-7 FPGA & \textbf{8.9}     & \textbf{2.34} &\textbf{1.52}& \textbf{$\sim 99.97$} \\ \hline
\end{tabular}
\vspace{2pt}
 		\caption{Overhead Comparison with literature}
 	\label{tab:lit}
 \end{table*} 
Table \ref{tab:compalgo} reports the slice, delay and power consumption of both unprotected (baseline) and protected versions with RESWO of the Kyber, CRYSTALS-Dilithium, Falcon, and NTRU algorithms.
Fig. \ref{fig:overhead} compares the overhead of Barrett Reduction with RESWO in terms of slice or area, delay and power consumption for the Kyber, CRYSTALS-Dilithium, Falcon, and NTRU PQC algorithms. It is observed that the resource and power consumption of different PQC algorithms vary depending on the values of $q$ and $n$. It is to be noted that the $Barrett$ $Reduction$ is already a subcomponent of $CT$-$BU$. For instance, in the first row of Table \ref{tab:compalgo}, $Kyber$ $Barrett$ is a subcomponent of $Kyber$ $CT$-$BU$ (baseline). The implementation cost of $Kyber$ $Barrett$ $with$ $RESWO$ includes both $Kyber$ $Barrett$ and the fault detection block $RESWO$. As shown in Fig. \ref{fig:overhead}, the overheads $OH$(\%) in the implementation cost of PQC algorithms (area/slice, delay, and power) for different CT-BU designs are measured using equ. \ref{equ:overhead}.
\begin{equation}
\label{equ:overhead}
\small{OH=\frac{PQC_i~Barrett~with~RESWO-PQC_i~Barrett}{PQC_i~CT-BU} \times 100}
\end{equation}
Here the $PQC_i$ has 4 options:  Kyber, CRYSTALS-Dilithium, Falcon, and NTRU.	The area. delay and power overheads are calculated using equ. \ref{equ:overhead}.
\begin{figure}[!htb]
\centering
\includegraphics[width=0.5\textwidth]{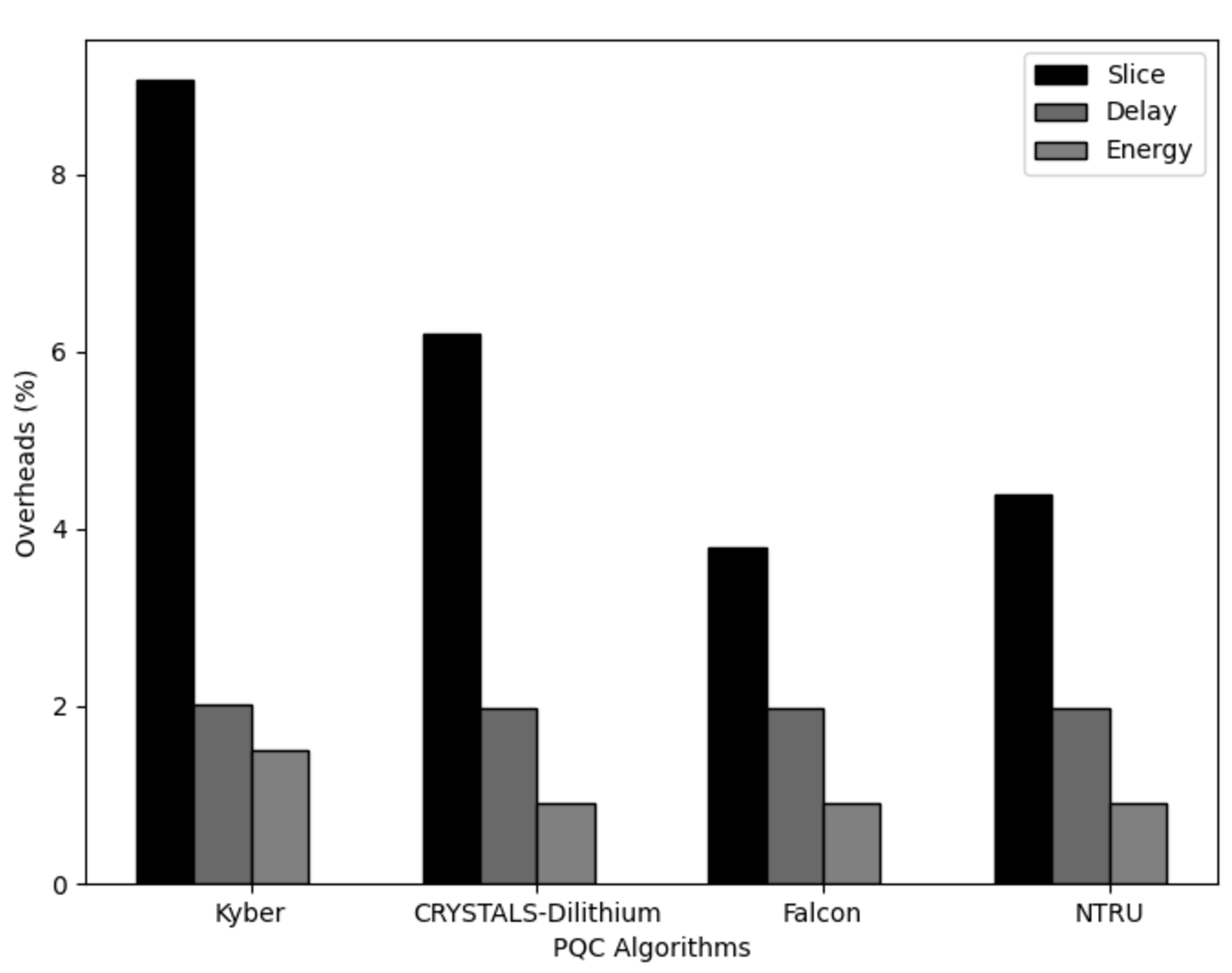}
\caption{Overheads of Proposed Barrett Reduction with RESWO in Different PQC Algorithms}
\label{fig:overhead}
\end{figure}
\subsubsection{Overheads of RESWO, RENO \& RESO in CT-BU of Kyber}
This paper implements the RESWO, RENO and RESO in the Barrett Reduction of $CT-BU$ placed inside Kyber and reports the slices, LUT. DSP, power and delay in Table \ref{tab:comp}.
It shows that the resource and power consumption of the proposed RESWO fault detection model are similar to RENO and RESO. However, RESWO shows lower delay compared to RENO and RESO.  The delay of RESWO is 9.51 ns, outperforming RENO (9.67 ns) and RESO (9.61 ns). The proposed RESWO, along with RENO and RESO, consumes 1.52\% more energy compared to the unprotected CT-BU.
 \begin{table}[!htbp]
 	\centering	
	\begin{tabular}{|>{\centering\arraybackslash}p{1.4cm}|>{\centering\arraybackslash}p{0.8cm}|>{\centering\arraybackslash}p{1.2cm}|>{\centering\arraybackslash}p{0.90cm}|>{\centering\arraybackslash}p{0.7cm}|>{\centering\arraybackslash}p{0.7cm}|}
	\hline
	\textbf{Block Names}& \textbf{Slices} & \textbf{LUTs} /\textbf{FFs} & \textbf{DSPs}/ \textbf{BRAMs}&\textbf{Power (mw)}& \textbf{Delay (ns)} \\ \hline
	CT BU (baseline)   &573& 972/ 239   & 2/1&131  & 9.39  \\ \hline
	Barrett    &76& 254/ 89   & 0/0 &101   & 7.177   \\ \hline
	RESWO    &52& 190/ 38   & 0/0& 2  & 9.51  \\ \hline
    RENO    &56& 197/ 48   & 0/0& 2  & 9.67  \\ \hline
    RESO    &51& 182/ 42   & 0/0& 2  & 9.61  \\ \hline
	\multicolumn{6}{|c|}{\textbf{Artix-7 (xc7a100tcsg324-3), n=256, q=3329, l=12, w=4,}}  \\ 
	\multicolumn{6}{|c|}{\textbf{clock=100MHz}}  \\ \hline
\end{tabular}
\vspace{2pt}
 		\caption{Overhead of RESWO, RENO and RESO}
 	\label{tab:comp}
 	\end{table}
\subsubsection{Overheads of RESWO, RENO \& RESO Compared to Other Fault Detection Techniques in Literature}
Table \ref{tab:lit} reports implementation cost and error detection efficiency of our RESWO, RENO and RESO schemes with existing solutions. To the best of our knowledge, there is no existing Kyber implementation with a fault detection method available in the literature. Therefore, Table \ref{tab:lit} presents a comparison of our Kyber-based fault detection methods with other fault detection methods used in different cryptographic algorithms. The proposed RESWO has a slightly higher area overhead compared to RESO and lower overhead compared to RENO. However, the proposed RESWO achieves the best delay compared to RESO and RENO. Unlike conventional techniques such as 1/2/3-bit parity or spatial duplication, which introduce substantial area and power overhead, RESWO achieves comparable fault coverage with lower resource consumption. The area overhead in Table \ref{tab:lit} is calculated by number slices. Area overhead marked by * in Table \ref{tab:lit} are approximated by equ. \ref{equ:sec}.
\begin{align}
\label{equ:sec}
SEC=&0.25 \times LUTs + 0.125 \times FFs + 100 \times DSPs\\
&+ 200 \times BRAMs \nonumber
\end{align}
Here $SEC$ is referred as Slice Effective Cost (SEC).
This approximation method is taken from \cite{liu}. Table \ref{tab:lit} shows that the overheads of the proposed RESWO, RENO, and RESO fault detection models in Barrett Reduction are reasonable and highly competitive with existing fault detection solutions for both PQC and classical cryptography. The higher error detection efficiency of our methods is acceptable considering this minimal implementation cost. $NA$ in Table \ref{tab:lit} refers to data that is $Not$ $Available$.
\subsection{Error Coverage}
To measure the error detection efficiency of the proposed RESWO, RENO, and RESO, these algorithms are implemented in Python and executed on an Ubuntu $24.04$ system with an $i5$ processor and 8 GB of RAM, utilizing 1.5 million samples. This process simulates fault injection in two ways: (i) Random fault injection, where fault bits are injected at random positions of the operands.
(ii) Burst fault injection, where fault bits are injected in consecutive positions of the operands.
This paper studies random and burst fault injection processes in three modes : i) Fault in $\alpha$ where random and burst faults are injected only in $\alpha$, ii) Fault in $\beta$ where random and burst faults are injected only in $\beta$ and iii) Fault in $\alpha$ and $\beta$ where random and burst faults are injected in both $\alpha$ and $\beta$.
Table \ref{tab:fault_reswo}, Table \ref{tab:fault_reno}, and Table \ref{tab:fault_reso} show that the fault detection efficiency of RESWO, RENO, and RESO ranges from 99.95\% to 99.97\% across different fault injection scenarios and fault modes. From Table \ref{tab:fault_reswo}, it is observed that $w$ has minimal effect on fault coverage of RESWO method. As $w$ does not affect error detection efficiency for RESWO, RENO, and RESO, the authors report the $w$ vs. error detection efficiency analysis only for RESWO and intentionally omit it for the RENO and RESO methods. However, changes in $w$ alter the bit size of all logic elements used in Barrett Reduction, which are computed within a single clock cycle. Therefore, $w$ significantly impacts the implementation cost, including slice utilization, power consumption and delay.

  \begin{table}[htbp!]
	\centering
	\begin{tabular}{|>{\centering\arraybackslash}p{0.1cm}|>{\centering\arraybackslash}p{0.6cm}|>{\centering\arraybackslash}p{0.75cm}|>{\centering\arraybackslash}p{0.65cm}|>{\centering\arraybackslash}p{0.75cm}|>{\centering\arraybackslash}p{0.65cm}|>{\centering\arraybackslash}p{0.75cm}|>{\centering\arraybackslash}p{0.65cm}|}
		\hline
		\textbf{w}&\textbf{\#} & \multicolumn{6}{c|}{\textbf{Fault Detection Efficiency(\%)}}\\ \cline{3-8}
		& \textbf{fault}  & \multicolumn{2}{c|}{ \textbf{fault in $\alpha$}} & \multicolumn{2}{c|}{\textbf{fault  in $\beta$}} & \multicolumn{2}{c|}{\textbf{fault  in $\alpha$ \& $\beta$}} \\ \cline{3-8}
		       & \textbf{bits($\eta$)}   & \textbf{random}&\textbf{burst} &\textbf{random}&\textbf{burst}   &\textbf{random}&\textbf{burst}  \\ \hline
            & 1    &  99.97  &    -      & 99.96 &   -        & 99.97  &    -  \\
			& 3    &  99.97  &  99.97    & 99.97 & 99.97      & 99.97  &  99.97   \\
			& 5    &  99.94  &  99.97    & 99.94 & 99.97      & 99.97  &  99.97   \\
		4	& 11   &  99.95  &  99.97    & 99.95 & 99.97      & 99.97  &  99.97   \\
			& 17   &  99.95  &  99.95    & 99.95 & 99.94      & 99.97  &  99.96   \\
			& 23   &  99.95  &  99.95    & 99.96 & 99.95      & 99.96  &  99.97 \\ 
			\hline
            & 1   & 99.97  &    -      & 99.97 &   -        & 99.97  &    - \\ 
			& 3   &   99.97  &  99.97    & 99.97 & 99.97      & 99.97  &  99.97   \\ 
			& 5   &  99.95  &  99.97    & 99.95 & 99.97      & 99.97  &  99.97 \\
		8	& 11   &   99.95  &  99.97    & 99.95 & 99.98      & 99.97  &  99.97  \\
			& 17   &   99.95  &  99.95    & 99.95 & 99.95      & 99.97  &  99.97  \\
			& 23   &   99.95  &  99.95    & 99.95 & 99.95      & 99.97  &  99.97   \\

		\hline

			& 1   & 99.97  &    -      & 99.97 &   -        & 99.97  &    - \\ 
			& 3   &   99.97  &  99.97    & 99.97 & 99.97      & 99.97  &  99.97   \\ 
			& 5   &  99.95  &  99.97    & 99.95 & 99.97      & 99.97  &  99.97 \\
		24	& 11   &   99.95  &  99.97    & 99.95 & 99.98      & 99.97  &  99.97  \\
			& 17   &   99.95  &  99.95    & 99.95 & 99.95      & 99.97  &  99.97  \\
			& 23   &   99.95  &  99.95    & 99.95 & 99.95      & 99.97  &  99.97   \\

		\hline
\multicolumn{8}{|c|}{\textbf{l=24, sample size=1.5 million}}\\\hline		
		
	\end{tabular}
	\vspace{2pt}
	\caption{Error Detecting Efficient for $n$ bit Random \& Burst Flipping using RESWO}
	\label{tab:fault_reswo}
	
\end{table}

 \begin{table}[htbp!]
	\centering
	\begin{tabular}{|>{\centering\arraybackslash}p{0.1cm}|>{\centering\arraybackslash}p{0.6cm}|>{\centering\arraybackslash}p{0.75cm}|>{\centering\arraybackslash}p{0.65cm}|>{\centering\arraybackslash}p{0.75cm}|>{\centering\arraybackslash}p{0.65cm}|>{\centering\arraybackslash}p{0.75cm}|>{\centering\arraybackslash}p{0.65cm}|}
		\hline
		\textbf{w}&\textbf{\#} & \multicolumn{6}{c|}{\textbf{Fault Detection Efficiency(\%)}}\\ \cline{3-8}
		& \textbf{fault}  & \multicolumn{2}{c|}{ \textbf{fault in $\alpha$}} & \multicolumn{2}{c|}{\textbf{fault  in $\beta$}} & \multicolumn{2}{c|}{\textbf{fault  in $\alpha$ \& $\beta$}} \\ \cline{3-8}
		       & \textbf{bits($\eta$)}   & \textbf{random}&\textbf{burst} &\textbf{random}&\textbf{burst}   &\textbf{random}&\textbf{burst}  \\ \hline
		    & 1    &  99.96  &    -      & 99.97 &   -        & 99.97  &    -  \\
			& 3    &  99.96  &  99.96    & 99.96 & 99.96      & 99.96  &  99.97   \\
			& 5    &  99.96  &  99.96    & 99.96 & 99.97      & 99.96  &  99.97   \\
		24	& 11   &  99.96  &  99.96    & 99.96 & 99.96      & 99.96  &  99.97   \\
			& 17   &  99.96  &  99.96    & 99.96 & 99.96      & 99.96  &  99.96   \\
			& 23   &  99.96  &  99.96    & 99.96 & 99.96      & 99.96  &  99.97 \\ 
			\hline
	
\multicolumn{8}{|c|}{\textbf{l=24, sample size=1.5 million}}\\\hline		
		
	\end{tabular}
	\vspace{1pt}
	\caption{Error Detecting Efficient for $n$ bit Random \& Burst Flipping using RENO}
	\label{tab:fault_reno}
	\vspace{-20pt}
\end{table}

 \begin{table}[htbp!]
	\centering
	\begin{tabular}{|>{\centering\arraybackslash}p{0.1cm}|>{\centering\arraybackslash}p{0.6cm}|>{\centering\arraybackslash}p{0.75cm}|>{\centering\arraybackslash}p{0.65cm}|>{\centering\arraybackslash}p{0.75cm}|>{\centering\arraybackslash}p{0.65cm}|>{\centering\arraybackslash}p{0.75cm}|>{\centering\arraybackslash}p{0.65cm}|}
		\hline
		\textbf{w}&\textbf{\#} & \multicolumn{6}{c|}{\textbf{Fault Detection Efficiency(\%)}}\\ \cline{3-8}
		& \textbf{fault}  & \multicolumn{2}{c|}{ \textbf{fault in $\alpha$}} & \multicolumn{2}{c|}{\textbf{fault  in $\beta$}} & \multicolumn{2}{c|}{\textbf{fault  in $\alpha$ \& $\beta$}} \\ \cline{3-8}
		       & \textbf{bits($\eta$)}   & \textbf{random}&\textbf{burst} &\textbf{random}&\textbf{burst}   &\textbf{random}&\textbf{burst}  \\ \hline
		    & 1    &  99.97  &    -      & 99.97 &   -        & 99.96  &    -  \\
			& 3    &  99.97  &  99.96    & 99.96 & 99.97      & 99.96  &  99.97   \\
			& 5    &  99.96  &  99.97    & 99.96 & 99.97      & 99.96  &  99.97   \\
		24	& 11   &  99.95  &  99.97    & 99.97 & 99.96      & 99.96  &  99.97   \\
			& 17   &  99.97  &  99.97    & 99.97 & 99.96      & 99.97  &  99.96   \\
			& 23   &  99.96  &  99.96    & 99.96 & 99.96      & 99.97  &  99.97 \\ 
			\hline
	
    \multicolumn{8}{|c|}{\textbf{l=24, sample size=1.5 million}}\\\hline		
		
	\end{tabular}
	\vspace{1pt}
	\caption{Error Detecting Efficient for $n$ bit Random \& Burst Flipping using RESO}
	\label{tab:fault_reso}
	\vspace{-20pt}
\end{table}

\section{Conclusion}
\label{sec:con}
This manuscript addresses the problem of fault detection in Barrett Reduction of CT-BU, which is the most critical and implementation-expensive hardware block in the latest PQC infrastructure. Natural faults or intentional faults induced during a side-channel attack on such fundamental hardware blocks may compromise the security of quantum attack-resistant PQC algorithms.
This manuscript presented a new recomputation based fault detection algorithm named RESWO for Barrett Reduction integrated within a CT-BU. We also compared the delay, resource utilization, power consumption and error coverage of RESWO with two existing recomputation-based fault detection schemes, RENO and RESO. The proposed RESWO has similar error detection efficiency, resource utilization, and power consumption, while it achieves lesser delay compared to RENO and RESO. To the best of our knowledge, this is the first time RENO and RESO are used in Barrett Reduction placed inside CT-BU. The proposed RESWO and other existing fault detection schemes used in our method are capable of addressing both permanent and transient faults. 
Although the fault detection schemes RESWO, RENO and RESO used in this paper is specifically designed for Barrett Reduction of CT-BU, it can also be adopted to any polynomial multiplication with modular reduction, where Barrett Reduction is used. The source code of this work is available on GitHub \footnote{https://github.com/rourabpaul1986/NTT/tree/master/barrett}. In the future, the authors will explore various fault detection algorithms for different hardware components of PQC algorithms.

\bibliographystyle{unsrt}  
\bibliography{IEEEexample}

\end{document}